\documentclass{article}
\usepackage[margin=1in]{geometry}
\usepackage[T1]{fontenc}
\usepackage[utf8]{inputenc}
\usepackage[all=normal,bibliography=tight]{savetrees}

\usepackage{listings}
\usepackage{cite}
  \usepackage{mathrsfs}

  \usepackage{amstext,amsfonts,amsthm,amsmath,amssymb}

\usepackage{graphicx,tikz}
\usepackage{comment}
\usepackage{url}
\usepackage{xspace}
\usepackage{todonotes}
\usepackage[shortlabels]{enumitem}
\usepackage[absolute]{textpos}

\def\cqedsymbol{\ifmmode$\lrcorner$\else{\unskip\nobreak\hfil
\penalty50\hskip1em\null\nobreak\hfil$\lrcorner$
\parfillskip=0pt\finalhyphendemerits=0\endgraf}\fi}

\newtheorem{lemma}{Lemma}[section]
\newtheorem{proposition}[lemma]{Proposition}

\newtheorem{theorem}[lemma]{Theorem}

\theoremstyle{definition}

\newcommand{\Oh}{\mathcal{O}}
\newcommand{\pmc}{\Omega}
\newcommand{\cc}{\mathtt{cc}}
\newcommand{\pmcfam}{\mathcal{F}}
\newcommand{\compfam}{\mathcal{D}}

\newcommand{\weight}{\mathfrak{w}}
\newcommand{\poten}{\zeta}

\title{On the Maximum Weight Independent Set Problem in graphs without induced cycles of length at least five}

\author{ 
  Maria Chudnovsky\thanks{Supported by NSF grants DMS-1763817. This material is based upon work supported in part by the U. S. Army Research Office under
    grant number  W911NF-16-1-0404.}\\
Princeton University, Princeton, NJ 08544 \and
 Marcin Pilipczuk\thanks{This research is a part of a project that has received funding from the European Research Council (ERC)
under the European Union's Horizon 2020 research and innovation programme
Grant Agreement no.~714704.} 
  \\ Institute of Informatics, University of Warsaw\\Banacha 2, 02-097 Warsaw, Poland \and 
Micha\l{} Pilipczuk\thanks{This research is a part of a project that has received funding from the European Research Council (ERC)
under the European Union's Horizon 2020 research and innovation programme
Grant Agreement no.~677651.} 
  \\ Institute of Informatics, University of Warsaw\\Banacha 2, 02-097 Warsaw, Poland \and 
    St\'{e}phan Thomass\'{e} \\ ENS de Lyon, 69364 Lyon Cedex 07, France
}

\date{}

\begin{document}

\maketitle

\begin{textblock}{20}(0, 13.0)
\includegraphics[width=40px]{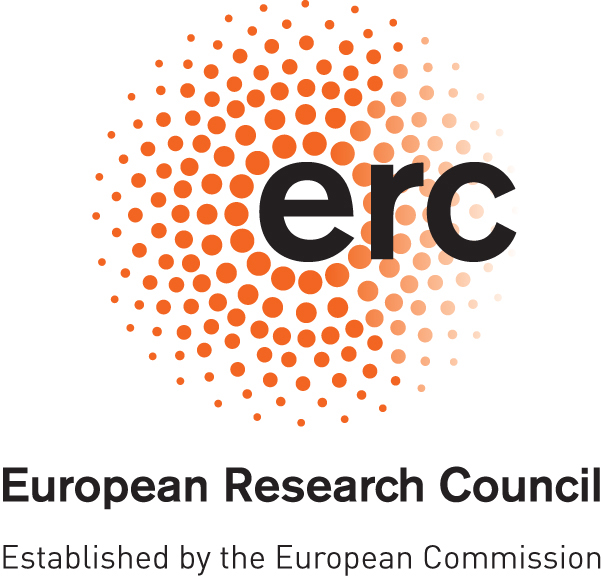}%
\end{textblock}
\begin{textblock}{20}(-0.25, 13.4)
\includegraphics[width=60px]{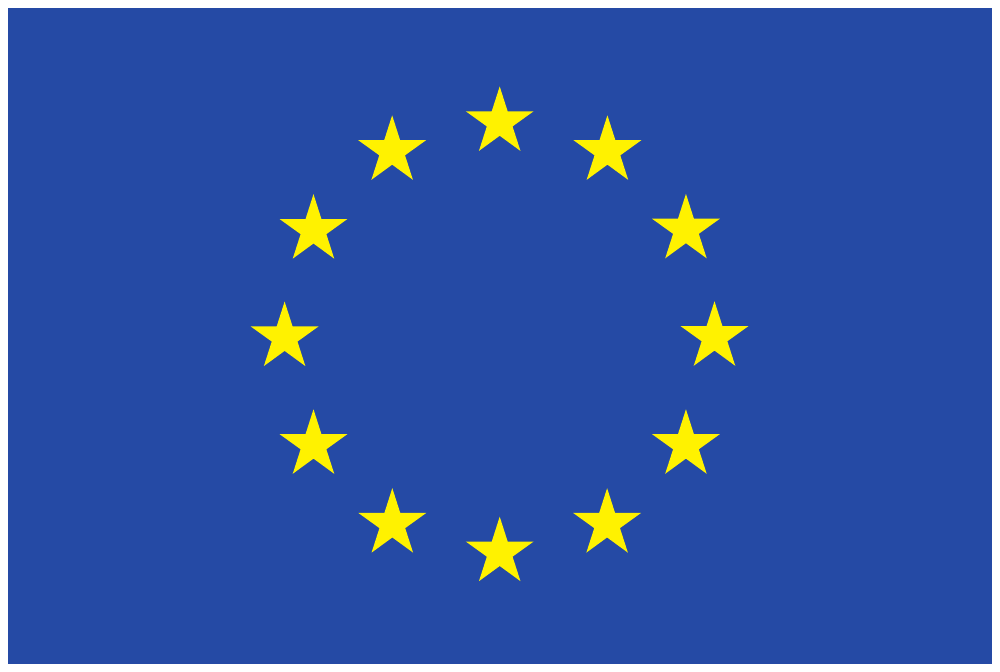}%
\end{textblock}

\begin{abstract}
  A {\em hole} in a graph is an induced cycle of length at least $4$, and an
  {\em  antihole} is the complement of an induced cycle of length at least $4$.
  A hole or antihole is {\em long} if its length is at least $5$.
  For an integer $k$, the {\em $k$-prism} is the graph consisting of two cliques
  of size $k$ joined by a matching.
The complexity of \textsc{Maximum (Weight) Independent Set} (\textsc{MWIS}) in long-hole-free graphs remains an important open problem. 
In this paper we give a polynomial time algorithm to solve \textsc{MWIS} in
long-hole-free graphs with no $k$-prism (for any fixed integer $k$), and
a subexponential algorithm for  \textsc{MWIS} in long-hole-free graphs
in general. As a special case this gives a polynomial time algorithm
to find a maximum weight clique in perfect graphs with no long antihole, and no hole of length $6$.
The algorithms use the framework of minimal chordal completions and potential maximal cliques.
\end{abstract}

\section{Introduction}
All graphs in this paper are finite and simple.
A {\em clique} in a graph is a set of pairwise adjacent vertices, and an
{\em independent set} (or a {\em stable set}) is a set of pairwise non-adjacent
vertices. The {\em chromatic number} $\chi(G)$ of a graph $G$
is the smallest number of independent sets of $G$ whose union is
$V(G)$.
A graph $G$ is \emph{perfect} if every induced subgraph $H$ of $G$
satisfies $\chi(H)=\omega(H)$, where $\chi(H)$ is the chromatic number
of $H$ and $\omega(H)$ is the maximum clique size in $H$.  In a graph
$G$, a \emph{hole} is an induced  cycle with at least four vertices
and an \emph{antihole} is the complement of a hole. The {\em length}
of a hole or an antihole is the number of vertices in it. A
hole or antihole is {\em long} if it has length at least $5$.

For two graphs $G$ and $F$ we say that $G$ \emph{contains} $F$ if $F$ is
isomorphic to an
induced subgraph of $G$.  A graph $G$ is {\em $F$-free} if it does not contain
$F$, and for a family of graphs ${\cal F}$, $G$ is {\em ${\cal
F}$-free} if $G$ is $F$-free for every $F\in {\cal F}$.  The class of perfect
graphs was introduced by Claude Berge {\cite{ber60}},
and became a class of central importance in graph theory.
Berge conjectured that a graph is perfect if and only if it does not contain
an odd hole or an odd antihole. This  question (the Strong
Perfect Graph Conjecture) was solved by Chudnovsky, Robertson, Seymour
and Thomas \cite{CRST}.
Moreover, Chudnovsky, Cornu\'ejols, Liu, Seymour and Vu\v{s}kovi\'c
\cite{CCLSV} devised a polynomial-time algorithm that determines if a
graph is perfect.

The \textsc{Maximum Independent Set} (\textsc{MIS})
is the problem of finding an independent set of maximum cardinality in a graph,
and the \textsc{Maximum Clique} (\textsc{MC}) is the problem of
finding a clique of maximum cardinality. Similarly, given a graph with
non-negative weights on its vertices, \textsc{Maximum Weight Independent Set} (\textsc{MWIS}) is  the problem of finding an independent set of maximum total weight, and \textsc{Maximum Weight Clique} (\textsc{MWC}) is the problem of finding a clique on maximum total weight.

It is known that the \textsc{Maximum Independent Set} (\textsc{MIS}),
\textsc{Maximum Weight Independent Set} (\textsc{MWIS}),
\textsc{Maximum Clique } (\textsc{MC}), and
\textsc{Maximum Weight Clique} (\textsc{MWC})  problems can be solved in 
 polynomial time on perfect graphs using the algorithm of Gr\"otschel, Lov\'asz and
Schrijver \cite{GLS}.  This algorithm however is not
combinatorial and uses the ellipsoid method.
We do not have a precise definition of a ``combinatorial algorithm''
(though a good  rule of thumb  is: an algorithm not using division),
but roughly  we mean an algorithm that can be described as a sequence of operations applied directly to vertices and edges of the graph in question, without using
techniques such as the ellipsoid method, the simplex method, etc.
No
combinatorial polynomial-time algorithm is known for any of the above problems
in perfect  graphs; finding one is a major open problem in the field.
At the moment we do not even have a polynomial-time
combinatorial algorithm to solve \textsc{MIS} in perfect graphs with no hole of length four.  Another important special case is the \textsc{MC} problem
for perfect graphs with no long antiholes. By taking complements, the latter
question
is a special case of solving \textsc{MIS} in the class of long-hole-free graphs.
At the time of submission of this paper the latter was an open problem.
It has since been solved by two of us and Tara Abrishami \cite{ACP}, however
that algorithm is conceptually much more complicated than the special case
presented in this paper. Also, since the $k$-prism has exponentially many
separators, the separator enumeration result presented in this paper is in  a
sense best possible.

We denote by $P_t$  the path on $t$ vertices. The {\em length} of a path is
the number of edges in it (thus the length of $P_t$ is $t-1$).
Recently significant progress on the question of the complexity of
\textsc{MWIS} was made using the approach of ``potential maximal cliques
`` (PMCs; we will give a precise defintion later in the paper) that was
originally developed by Bouchitt\'{e} and
Todinca~\cite{BouchitteT01,BouchitteT02}. A milestone result was obtained in 2014 by Lokshtanov, Vatshelle, and Villanger~\cite{LokshtanovVV14}
who designed a polynomial-time algorithm for \textsc{MWIS} in $P_5$-free graphs.
Within the same framework, recently Grzesik et al.~\cite{GrzesikKPP19} showed polynomial-time algorithm for \textsc{MWIS} in $P_6$-free graphs.

The starting point of this paper was to try and apply this powerful technique
to various subclasses of perfect graphs. However, our main results are 
about a class of graphs that contains both perfect and imperfect graphs,
and contains an interesting subclass of perfect graphs,
as follows. For an integer $k>0$ the {\em $k$-prism} is the graph
consisting of
two cliques of size $k$, and a $k$-edge matching between them.
More precisely, the $k$-prism $G$ has   vertex set
$\{a_1, \ldots, a_k, b_1, \ldots, b_k\}$;  each of
the sets $\{a_1, \ldots, a_k\}$ and $\{b_1, \ldots, b_k\}$ is a clique,
$a_ib_i \in E(G)$ for every $i \in \{1, \ldots, k\}$, and there are no other
edges in $G$. Our first result is the following:

\begin{theorem}\label{thm:kprism-alg}
For every integer $k>0$ 
the \textsc{Maximum Weighted Independent Set} problem in a
(long-hole, $k$-prism)-free
$n$-vertex graph $G$ can be solved in time $n^{\Oh(k)}$.
\end{theorem}

We do not see how to turn Theorem~\ref{thm:kprism-alg} into an FPT-algorithm
(withouth using the resut of \cite{ACP}).
Since  prisms (which is a structure more general than a $k$-prism, and we will not define it here) come up naturally in the context of perfect graphs \cite{CRST}, the following corollary, obtained by taking complements, is of interest:

\begin{theorem}\label{thm:kprism-perfect}
Let $k>0$ be an integer.
Let $G$ be an $n$-vertex perfect graph with no long antihole, and such that
the complement of $G$ does not contain the $k$-prism.
Then the  \textsc{Maximum Weighted Clique} problem  in $G$
can be solved in time $n^{\Oh(k)}$.
In particular, the \textsc{Maximum Weighted Clique} problem in a perfect
$n$-vertex graph $G$ with no long antihole and no hole of length $6$
can be solved in time $n^{\Oh(1)}$.
\end{theorem}
The last statement of Theorem~\ref{thm:kprism-perfect} follows from the fact that the complement of a the $3$-prism is the cycle of length $6$.

The algorithm of Theorem~\ref{thm:kprism-alg} easily implies a subexponential algorithm for \text{MWIS} in long-hole-free graphs, as we now explain.
\begin{theorem}\label{thm:subexp1}
The \textsc{Maximum Weighted Independent Set} problem in a long-hole-free $n$-vertex graph $G$
can be solved in time $n^{\Oh(\sqrt{n})}$.
\end{theorem}
\begin{proof}
Set $k = \lfloor \sqrt{n} \rfloor$ and check (by exhaustive enumeration) if $G$ contains the $k$-prism as an induced subgraph.
If such a prism $P$ has been found, then branch into $\Oh(n)$ subcases guessing $V(P) \cap I$ for the sought
optimum independent set $I$ (since $P$ consists of two cliques, it intersects with any independent set in at most two vertices and
$|V(P)| = 2\lfloor \sqrt{n} \rfloor $).
In every branch, delete from the graph $V(P) \cup N(V(P) \cap I)$ for the guessed value of $V(P) \cap I$ and recurse;
since $|V(P)| \geq 2k$, the number of vertices in the graph drops by at least $2\lfloor \sqrt{n} \rfloor$. 
Otherwise, if no such $P$ is found, apply the algorithm of Theorem~\ref{thm:kprism-alg}, which now runs in time
$n^{\Oh(\sqrt{n})}$. 
Standard analysis shows that this algorithm has running time bound $n^{\Oh(\sqrt{n})}$. 
\end{proof}

Recently, two groups of authors~\cite{BacsoLMPTL19,GORSSS18} reported a subexponential-time algorithm for \textsc{MWIS} in a related
class of $P_t$-free graphs for every fixed $t$.  Their result depends
heavily on the notion of ``bounded balanced separators'', which we
explain next.
A balanced separator for a
graph $G$ and a weight function $\weight : V(G) \to [0,+\infty)$ is 
a set of vertices $X \subseteq V(G)$ such that every connected component $C$ of $G-X$ has total weight (w.r.t. $\weight$)
at most half of the total weight of $V(G)$. We say that a graph class $\mathcal{G}$ has \emph{balanced separators bounded
by $f$} if for every $G \in \mathcal{G}$ and every weight function $\weight : V(G) \to [0,+\infty)$ there exists
a balanced separator for $G$ and $\weight$ of size at most $f(G)$. 
The main technical statement of~\cite{BacsoLMPTL19} is that a $P_t$-free graph $G$ admits a balanced separator of size
bounded by $(t-1)\Delta(G) +1$ where $\Delta(G)$ is the maximum degree in $G$.
Our second result is a similar statement for long-hole-free graphs, which we believe is of independent interest.
\begin{theorem}\label{thm:seps}
For every long-hole-free graph $G$ and every weight function $\weight : V(G) \to [0,+\infty)$
there exists a balanced separator of $G$ and $\weight$ of size at most $3(\Delta(G)+1)$.
\end{theorem}
Standard arguments (see e.g.~\cite{BacsoLMPTL19}) show that if a graph class $\mathcal{G}$ has balanced separators bounded by $f$
then the treewidth of a graph $G \in \mathcal{G}$ is bounded by $\Oh(f(G))$ and, if a balanced separator of size at most $f(G)$
for given $G$ and $\weight$ can be found in polynomial time, so can a tree decomposition of width $\Oh(f(G))$.
In~\cite{BacsoLMPTL19}  a subexponential algorithm for \textsc{MWIS} in a $P_t$-free $n$-vertex graph 
with running time bound $2^{\Oh(\sqrt{t n \log n})}$ is obtained by 
first setting a threshold $\tau = \sqrt{n \log n / t}$, branching exhaustively on vertices of degree at least $\tau$
and, once the maximum degree drops below this threshold, by computing a tree decomposition of width $\Oh(\sqrt{t n \log n})$
and solving \textsc{MWIS} by a dynamic programming algorithm on this tree decomposition.
Following exactly the same strategy with threshold $\tau = \sqrt{n \log n}$ we obtain the following.
\begin{theorem}\label{thm:subexp2}
The \textsc{Maximum Weighted Independent Set} problem on a long-hole $n$-vertex graph $G$
can be solved in time $2^{\Oh(\sqrt{n \log n})}$.
\end{theorem}

\paragraph{Organization}
In Section  \ref{sec:prelims} we explain the general framework of potential
maximal cliques. In Section \ref{sec:kprism} we prove  
Theorem~\ref{thm:kprism-alg}, and finally in
Section~\ref{sec:3dom} we prove  Theorem~\ref{thm:seps}.

\section{Separators and potential maximal cliques}\label{sec:prelims}

Let $G$ be a graph and let $X \subseteq V(G)$. We denote by $G[X]$ the sugbraph of $G$ induced by $X$ and by $G-X$ the graph $G[V(G) \setminus X]$.
A {\em component} of $X$ (or of $G[X]$)  is the vertex
set of a maximal connected subgraph of $G[X]$. We write $\cc(G)$ to mean the set
of connected components of $G$. We denote by $N(X)$ the set of vertices
of $V(G) \setminus X$ with a neighbor in $X$, and  write $N[X]=N(X) \cup X$.
When $X=\{x\}$  we use the notation $N(x)$ (or $N[x]$) instead of $N(\{x\})$
(or $N[\{x\}]$). For $u,v \in V(G)$ a {\em path from $u$ to $v$  via $X$}
is a path with ends $u,v$ and such that all of its internal vertices
belong to $X$. Observe that if $uv \in E(G)$ and $u,v \not \in X$,
the a path from $u$ to $v$ via $X$ is disjoint from $X$.
For $F \subseteq \binom{V(G)}{2} \setminus E(G)$ the graph $G+F$ has vertex set $V(G)$ and edge set $E(G) \cup F$.

A graph is {\em chordal} if it has no holes.
A set $F \subseteq \binom{V(G)}{2} \setminus E(G)$ is a \emph{fill-in}
or a \emph{chordal completion} (of $G$) 
if $G+F$ is a chordal graph. A fill-in $F$ is \emph{minimal} if it is
inclusion-wise minimal.

Let $X \subseteq V(G)$. For $s,t \in V(G) \setminus X$, we say that $X$
is an \emph{$s,t$-separator} if $s$ and $t$ lie in different connected
components of $G-X$. An $s,t$-separator is a \emph{minimal $s,t$-separator} if it is an inclusion-wise minimal $s,t$-separator. $X$ is said to be a 
\emph{minimal separator} if there exist $s,t \in V(G)$ such that $X$
is a minimal $s,t$-separator in $G$. We say that
$D \in \cc(G -X)$ is a {\em full component for $X$} if $N(D)=X$.
It is easy to see  that:

\begin{lemma}\label{twofull}
$X$  is a minimal separator if and only if at least
two members of $\cc(G-X)$ are full components.
\end{lemma}

An important property of minimal separators is that no new minimal separator appears when a minimal fill-in is added to a graph. More precisely:
\begin{proposition}[\cite{BouchitteT01}]
  Let $G$ be  a graph and let $F$ be a minimal fill-in for $G$. If $X$
  is a minimal separator of $G+F$, then $X$ is a minimal separator of $G$.
    Furthermore, $\cc(G+F-X)=\cc(G-X)$.
\end{proposition}

A set $\pmc \subseteq V(G)$ is a \emph{potential maximal clique} (PMC) if there exists a minimal fill-in $F$ of $G$
such that $\pmc$ is a maximal (inclusion-wise) clique of $G+F$. A PMC is
surrounded by minimal separators in the following sense:

\begin{proposition}[\cite{BouchitteT01}]\label{prop:pmcnbrs}
  Let $G$ be a graph, let $\pmc \subseteq V(G)$  be a PMC of $G$, and let
  $D \in \cc(G-\pmc)$. Then $N(D)$ is 
  a minimal separator of $G$ and $D$ is a full component for $N(D)$.
\end{proposition}

Next we state an important characterization of  PMCs in graphs.
\begin{theorem}[\cite{BouchitteT01}]\label{thm:pmc}
A set $\pmc \subseteq V(G)$ is a PMC in $G$ if and only if the following two conditions hold:
\begin{enumerate}
\item for every $D \in \cc(G-\pmc)$ we have $N(D) \subsetneq \pmc$;
\item for every $x, y \in \pmc$ either $x=y$, $xy \in E(G)$, or there exist $D \in \cc(G-\pmc)$ with $x,y \in N(D)$.
\end{enumerate}
\end{theorem}
In the second condition of Theorem~\ref{thm:pmc}, we say that a component $D$ \emph{covers} the nonedge $xy$.

Our main algorithmic engine is the following.

\begin{theorem}[\cite{BouchitteT01}]\label{thm:engine}
Given a graph $G$ with vertex weights and a family $\pmcfam$ that contains all
PMCs of $G$,  one can solve \textsc{MWIS} in $G$
in time polynomial in the size of $G$ and $\pmcfam$.
\end{theorem}

Thus it is enough to construct a family as in Theorem~\ref{thm:engine}. However,
it turns out that instead of constructing a family of PMCs, it is easier to
construct a family of components that result from deleting PMCs. This approach
was taken in~\cite{BouchitteT01,LokshtanovVV14,GrzesikKPP19}, and is justified by the following result:

\begin{theorem}[\cite{BouchitteT01}]\label{thm:pmc-lift-all}
  Given a graph $G$ and a family $\mathcal{G}$ of vertex sets of connected
  induced subgraphs of $G$
such that for every potential maximal clique $\pmc$ of $G$ we have $\cc(G-\pmc) \subseteq \mathcal{G}$,
one can compute the family $\pmcfam$ of all potential maximal cliques of $G$.
The running time of the algorithm and the size of the family $\pmcfam$ is bounded polynomially in the size of
$G$ and $\mathcal{G}$.
\end{theorem}
Our final observation is the following
\begin{theorem}\label{thm:sepsenough}
Given a graph $G$ and a family $\mathcal{S}$ of all minimal separators of $G$,
one can compute the family $\mathcal{G}$ of subsets of $V(G)$
such that for every potential maximal clique $\pmc$ of $G$ we have $\cc(G-\pmc) \subseteq \mathcal{G}$.
The running time of the algorithm and the size of the family $\mathcal{G}$ is bounded polynomially in the size of $G$ and $\mathcal{S}$.
\end{theorem}

  \begin{proof}
    For every $X \in \mathcal{S}$ we can compute in polynomial time the set
    $\cc(G-X)$. Let $\mathcal{G}=\bigcup_{X \in \mathcal{S}}\cc(G-X)$;
        we claim that $\mathcal{G}$ is the desired
    family.
    To see this, let $\pmc$ be a potential maximal clique of $G$ and let
    $D \in \cc(G-\pmc)$.
    By Proposition \ref{prop:pmcnbrs}, $N(D)$ is a minimal separator of
    $G$, and therefore  $D \in \cc(G-N(D)) \subseteq \mathcal{G}$.
  \end{proof}
We remark that all minimal separators in a graph can be enumerated in time polynomial in the graph size and the
number of output minimal separators~\cite{BouchitteT02}. In view of
Theorem~\ref{thm:sepsenough} from now on we focus on studying minimal separators.

\section{$k$-prism and minimal separators}\label{sec:kprism}
The goal
of this section is to prove Theorem~\ref{thm:kprism-alg}.
We say that a graph class $\mathcal{C}$ has the {\em polynomial separator property} if there exists
$b_{\mathcal{C}}$ such that every $G \in \mathcal{C}$ has at
most $|V(G)|^{b_{\mathcal{C}}}$ minimal separators.
In view of the results of Section~\ref{sec:prelims}, \textsc{MWIS} can be solved
in polynomial time in any graph class with the polynomial separator property.
It is easy to see that the  $k$-prism has $2^k-2$ minimal separators while being long-hole-free,
and therefore the class of long-hole-free graphs does not have the
polynomial separator property. In this section we prove that in long-hole-free
graphs $k$-prisms are the only reason the property is violated. 

We show:

\begin{theorem}\label{thm:kprism-minseps}
Let $k \geq 2$ be an integer and let $G$ be a long-hole-free graph that does not contain a $k$-prism.
Then $G$ has at most $|V(G)|^{k + 2}$ minimal separators.
\end{theorem}

We start with a lemma.

\begin{lemma}\label{lem:kprism}
Let $k \geq 2$ be an integer and let $G$ be a long-hole-free graph that does not contain a $k$-prism.
Let $S$ be a minimal separator in $G$ and let
$A, B \in \cc(G-S)$ with  $N(A) = N(B) = S$.
Assume that there exists $v \in A$ such that for every $A' \in \cc(G[A]-\{v\})$ we have $S \setminus N(A') \neq \emptyset$.
Then there exists a set $Z \subseteq A \cap N[v]$ such that
$v \in Z$, $|Z| \leq k$ and  $S \subseteq N(Z)$.
\end{lemma}
\begin{proof}
Let $S' = S \setminus N(v)$. We claim that we can enumerate the components of
$\cc(G[A]-\{v\})$ as $A_1, A_2, \ldots, A_m$ such that 
	$$N(A_1) \cap S' \supseteq N(A_2) \cap S' \supseteq \ldots \supseteq N(A_m) \cap S'.$$

Let $A^1$ and $A^2$ be two distinct connected components of $G[A]-\{v\}$; see Figure~\ref{fig:L32-1} for an illustration.
Assume that for $i=1,2$ there exist $v^i \in S' \cap N(A^i) \setminus N(A^{3-i})$, that is,
the neighborhoods of the components $A^1$ and $A^2$ are incomparable inside $S'$.
Let $P^i$ be a shortest path from $v^i$ to $v$ via $A^i$, and let $Q$ be a shortest path from $v^1$ to $v^2$ via $B$ (possibly $Q$ consists of one edge if $v^1$ and $v^2$ are adjacent).
Note that $P^1$ and $P^2$ are of length at least two as $v^1, v^2 \in S'$ (and
therefore $v^1,v^2$ are non-adjacent to $v$), while $Q$ is of length at least
one. Hence, the concatenation of $P^1$, $P^2$, and $Q$ is a hole of length at least five, a contradiction.  This proves the claim.

By the assumption of the Lemma, we know that
$S \setminus N(A_1) \neq \emptyset$; let $p \in S \setminus N(A_1)$.
Since $S=N(A) = N(A_1) \cup N(v)$, we deduce that $S' \subseteq N(A_1)$ and $pv \in E(G)$.

Assume there exists $x \in S'$ such that a shortest path $P$ from $x$ to $v$ via $A_1$ is of length at least three; see Figure~\ref{fig:L32-2} for an illustration.
Let $Q$ be a shortest path from $x$ to $p$ via $B$.
Then, the concatenation of $P$, $Q$, and the edge $pv$ is a hole of length at least five, a contradiction. We deduce that for every $x \in S'$, there exists a vertex $y \in N(v) \cap A_1$
with $xy \in E(G)$.
In particular, $S' \subseteq N(N(v) \cap A_1)$.

Consider now an inclusion-wise minimal set $Z' \subseteq A_1 \cap N(v)$ with $S' \subseteq N(Z')$.
By minimality, for every $z \in Z'$ pick $f(z) \in S' \setminus N(Z' \setminus \{z\})$; clearly, $zf(z) \in E(G)$.
Pick two distinct $z^1, z^2 \in Z'$; see Figure~\ref{fig:L32-3} for an illustration. Let $Q$ be a shortest path from $f(z^1)$ to $f(z^2)$ via $B$.
Let $P$ equal the edge $z^1z^2$ if it is present, or the concatenation of edges $vz^1$ and $vz^2$ if $z^1z^2 \notin E(G)$.
Then, the concatenation of $P$, $Q$, and edges $z^if(z^i)$ for $i=1,2$ is a hole of length at least five unless
both $z^1z^2 \in E(G)$ and $f(z^1)f(z^2) \in E(G)$.

Write  $f(Z') = \{f(z)~|~z \in Z'\}$ and  note that $|f(Z')|>|Z'|$.
We conclude  that both $Z'$ and $f(Z')$ are cliques of $G$.
Since $G$ does not contain a $k$-prism, it follows that $|Z'| < k$,
and therefore the set $Z = Z' \cup \{v\}$ has the desired properties.
This completes the proof.
\end{proof}

\begin{figure}[tb]
\begin{center}
\includegraphics{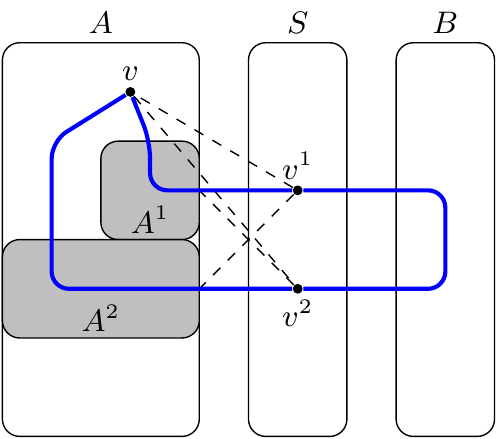}
\caption{First step in the proof of Lemma~\ref{lem:kprism}. If two components $A^1$ and $A^2$ have incomparable neighborhoods in $S'$, then
  a blue long hole appears.}\label{fig:L32-1}
\end{center}
\end{figure}

\begin{figure}[tb]
\begin{center}
\includegraphics{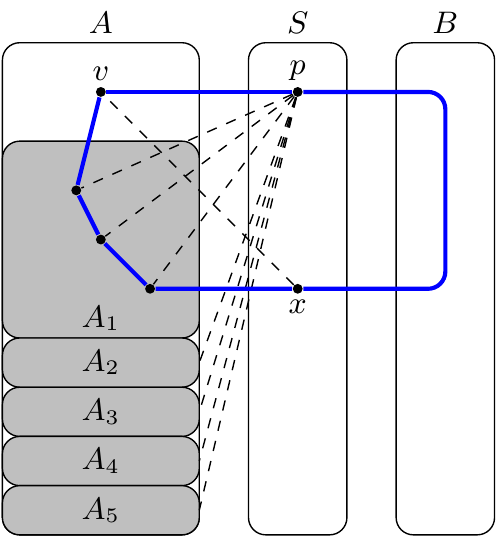}
\caption{Second step in the proof of Lemma~\ref{lem:kprism}. If a path from $v$ to $x \in S'$ through $A_1$ is too long, a blue long hole appears.}\label{fig:L32-2}
\end{center}
\end{figure}

\begin{figure}[tb]
\begin{center}
\includegraphics{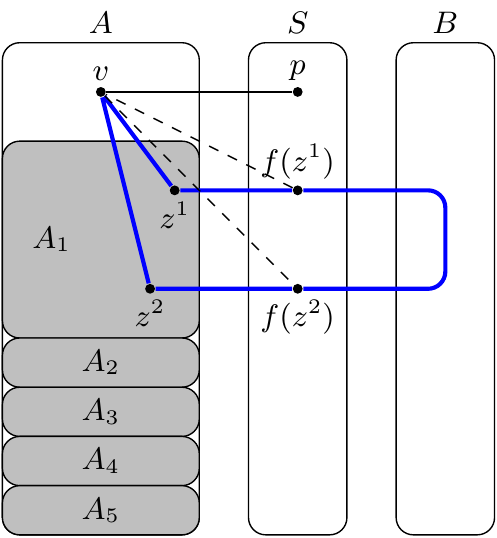}
\caption{Third step in the proof of Lemma~\ref{lem:kprism}. Unless both $z^1z^2 \in E(G)$ and $f(z^1)f(z^2) \in E(G)$, the blue cycle can be shortened to a long hole.}\label{fig:L32-3}
\end{center}
\end{figure}

Now we can prove Theorem~\ref{thm:kprism-minseps}.
\begin{proof}
For a set $S \subseteq V(G)$, we define $\poten_G(S) = \max(0, |\{A \in \cc(G-S)~|~N_G(A)=S\}| - 1)$.
Note that $\poten_G(S) > 0$ if and only if $S$ is a minimal separator. 
By induction on the number of vertices of $G$, we show that 
\begin{equation}\label{eq:kprism1}
\sum_{S \subseteq V(G)} \poten_G(S) \leq |V(G)|^{k + 2}.
\end{equation}
The statement is straightforward for $|V(G)| \leq 2$.

Pick an arbitrary $v \in V(G)$ and let $G' = G-\{v\}$.
Since by the inductive hypothesis \eqref{eq:kprism1} holds for $G'$, in order to
 show~\eqref{eq:kprism1}, it suffices to show that
\begin{equation}\label{eq:kprism2}
\sum_{S \subseteq V(G)} \poten_G(S) - \sum_{S \subseteq V(G')} \poten_{G'}(S) \leq |V(G)|^{k+1}.
\end{equation}
Let $S \subseteq V(G)$ be such that $\poten_G(S) > 0$, that is, $S$ is a minimal separator in $G$. By Lemma~\ref{twofull} at least
two members of $\cc(G-X)$ are full components.
We consider two cases.

We say that $S$ is \emph{special} if 
\begin{itemize}
\item $v \notin S$,
\item  if $A$ is the connected component of $G-S$ that contains $v$,
then $N_G(A) = S$, and
\item for every $A' \in \cc(G[A]-\{v\})$ we have $S \setminus N_G(A') \neq \emptyset$.
\end{itemize}
If $S$ is special, then by Lemma~\ref{lem:kprism} there exists $Z \subseteq A$ of size at most $k$ such that 
$S \subseteq N_G(Z)$. Thus, every connected component $B \in \cc(G-S)$ distinct from $A$
with $N_G(B) = S$ is a connected component of $G-N[Z]$. Since there are at most $|V(G)|^k$ choices
for $Z$, we infer that the contribution to the sum $\sum_{S \in V(G)} \poten_G(S)$
from the sets $S$ that are special is at most $|V(G)|^{k+1}$.

 It is now enough to obtain an upper bound on the contribution of the
non-speical separators of $G$ to the left hand side of
\eqref{eq:kprism2}.
Define $\poten'_G(S) = 0$ if $S$ is special and $\poten'_G(S) = \poten_G(S)$ otherwise. It suffices to show that:

\begin{equation}\label{eq:kprism22}
\sum_{S \subseteq V(G)} \poten'_G(S) \leq \sum_{S \subseteq V(G')} \poten_{G'}(S).
\end{equation}

If $S$ is a minimal separator that is not special, then either $v \in S$ or $v \notin S$ and,
   if $A$ is the connected component of $G-S$ that contains $v$,
then either $N_G(A) \subsetneq S$ or there is one connected component $A'$ of $G[A]-\{v\}$ that satisfies $N_{G'}(A') = S$.
In both cases $S' := S \setminus \{v\}$ is a minimal separator in $G'$.
Hence, to show~\eqref{eq:kprism2} it suffices to show that for every minimal separator $S'$ in $G'$
it holds that:
\begin{equation}\label{eq:kprism3}
\poten_{G'}(S') \geq \poten'_{G}(S') + \poten_G(S' \cup \{v\}).
\end{equation}

Note that $\poten_G(S' \cup \{v\})=\poten'_G(S' \cup \{v\})$ since
a minimal separator containing $v$ is not special.

Let $\mathcal{A} := \{A \in \cc(G'-S')~|~N_{G'}(A) = S'\}$ and let $\mathcal{B} := \{A \in \mathcal{A}~|~v \in N_G(A)\}$.
Clearly,
	$$\poten_{G'}(S') = \max(0, |\mathcal{A}| - 1).$$
If $\mathcal{B} \neq \emptyset$, then there exists a single connected component of $G-S'$ that contains $v$
	and all connected components of $\mathcal{B}$. Hence, 
	\begin{align*}
		\poten_{G}(S' \cup \{v\}) &= |\mathcal{B}|-1,\\
		\poten'_G(S') \leq \poten_{G}(S') &= |\mathcal{A}|-|\mathcal{B}|=\poten_{G'}(S')-\poten_{G}(S' \cup \{v\}).
	\end{align*}
This proves~\eqref{eq:kprism3} in the case $\mathcal{B} \neq \emptyset$.
Otherwise, if $\mathcal{B} = \emptyset$, then $S' \cup \{v\}$ is not a minimal separator in $G$ and $\poten_{G}(S' \cup \{v\}) = 0$.
If the connected component $A$ of $G-S'$ that contains $v$ satisfies $N(A) = S'$, then $S'$ is special in $G$, $\poten'_G(S') = 0$, and~\eqref{eq:kprism3} is proven.
Otherwise, we observe that $\poten_G(S') = \poten_{G'}(S')$ and we are done.
This completes the proof.
\end{proof}

Finally, we prove  Theorem~\ref{thm:kprism-alg}.
\begin{proof}[Proof of Theorem~\ref{thm:kprism-alg}.]
Since $G$ is $k$-prism-free,  Theorem~\ref{thm:kprism-minseps}
implies that the number of minimal separators in $G$ is at most
$|V(G)|^{k+2}$. By a result of \cite{BouchitteT02}, all
minimal separators of $G$ can be enumerated in time $n^{\Oh(k)}$.
Now Theorem~\ref{thm:kprism-alg} follows from Theorems~\ref{thm:sepsenough},~\ref{thm:pmc-lift-all}, and~\ref{thm:engine}.
\end{proof}

\section{Dominating a PMC with three vertices}\label{sec:3dom}
In this section we prove Theorem~\ref{thm:seps}. To this end, we show that in a long-hole-free graph $G$
every PMC is contained in a neighborhood of at most three vertices of $G$.
This is done by a sequence of structural lemmas that follows next. 

\begin{lemma}\label{lem:lh:minsep-dom-indset}
Let $G$ be a long-hole-free graph, let $S$ be a minimal separator in $G$, and let $A \in \cc(G-S)$ satisfy $N(A) = S$.
Then for every independent set $M \subseteq S$ there exists a vertex $a \in A$ with $M \subseteq N(a)$.
\end{lemma}
\begin{proof}
For every $a \in A$ let $f(a) = N(a) \cap M$.
Suppose that $M \setminus f(a) \neq \emptyset$ for every $a \in A$.
It follows that we can choose $a,a' \in A$ with the following properties:
\begin{itemize}
\item $f(a)$ is incluison-wise maximal among $\{f(a)
| a \in A\}$;
\item $f(a') \not \subseteq f(a)$; and
\item subject to the first two conditions, the length of a  shortest path from
  $a$ to $a'$ in $G[A]$ is the smallest possible.
\end{itemize}
Let $m' \in f(a') \setminus f(a)$. By the maximality of $f(a)$, there exists
$m \in f(a) \setminus f(a')$; see Figure~\ref{fig:L41}.
Let $P$ be a shortest path from $a$ to $a'$ in $G[A]$. Then $|V(P)|>1$.
By the second and third conditions above, 
for every $p \in V(P) \setminus \{a'\}$ we have that $f(p) \subseteq f(a)$.
In particular, no vertex of $V(P) \setminus \{a'\}$ is adjacent to $m'$.
Let $p$ be the neighbor of $m$ closest to $a'$ along $P$.
Then $p \neq a'$. Now $R=m-p-P-a'-m'$ is an induced path with at least
four vertices. By Lemma~\ref{twofull} there exists $B \neq A$ such that $B$ is  a full component for $S$, and  let $Q$ be a shortest
path with endpoints $m$ and $m'$ and all internal vertices in $B$. By concatenating
$R$ and $Q$ we get a hole of length at least five, a contradiction.
\end{proof}

\begin{figure}[tb]
\begin{center}
\includegraphics{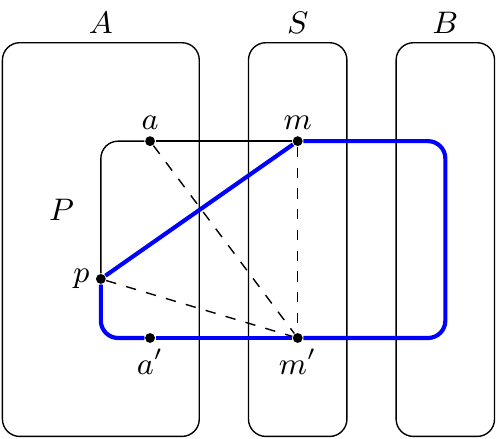}
\caption{Proof of Lemma~\ref{lem:lh:minsep-dom-indset}.}
  \label{fig:L41}
\end{center}
\end{figure}

\begin{lemma}\label{lem:lh:minsep-xab}
  Let $G$ be long-hole-free graph and let $S$ be a minimal separator in $G$.
Let $A, B \in \cc(G-S)$ with $N(A)=N(B)=S$.
Then for every $x \in S$ there exist $a \in N(x) \cap A$ and $b \in N(x) \cap B$ with $S \subseteq N[x] \cup N(a) \cup N(b)$.
\end{lemma}
\begin{proof}
%
For each $z\in S\setminus N[x]$, let $f(z)=N(z) \cap (A \cup B) \cap N(x)$.
Let $Z_0$ be the set of all $z \in S \setminus N[x]$ for which $f(z)$ is inclusion-wise minimal among $\{f(z)\colon z\in S \setminus N[x]\}$,
and let $Z \subseteq Z_0$ be an inclusion-wise maximal independent subset of $Z_0$.
Since $Z \subseteq S \setminus N[x]$, the set $Z \cup \{x\}$ is independent as well.
By Lemma~\ref{lem:lh:minsep-dom-indset}, there exists $a \in A$ and $b \in B$ with $Z \cup \{x\} \subseteq N(a) \cap N(b)$.

We claim that
\begin{equation}\label{eq:lh:minsep-xab-a0b0}
Z_0 \subseteq N(a) \cup N(b).
\end{equation}
Assume the contrary: there exists $z_0 \in Z_0 \setminus (N(a) \cup N(b))$. 
By the maximality of $Z$, there exists $z \in N(z_0) \cap Z$; see Figure~\ref{fig:L42}. By the choice of $a$ and $b$, we have
$a,b \in (N(z) \setminus N(z_0)) \cap N(x)$. Since $z \in Z_0$, it follows
from the defintion of $Z_0$ that
 the set $f(z)$ is inclusion-wise minimal among
 $\{f(z)\colon z\in S \setminus N[x]\}$, and therefore
there exists a vertex
$v \in f(z_0)\setminus f(z)=(N(z_0) \setminus N(z)) \cap (A \cup B) \cap N(x)$.
By symmetry, assume $v \in A$. 
Then $x-v-z_0-z-b-x$ is a long hole in $G$, a contradiction. 
This finishes the proof of~\eqref{eq:lh:minsep-xab-a0b0}.

By the definition of $Z_0$,~\eqref{eq:lh:minsep-xab-a0b0} implies $S \setminus N[x] \subseteq N(a) \cup N(b)$. This finishes the proof of the lemma.
\end{proof}

\begin{figure}[tb]
\begin{center}
\includegraphics{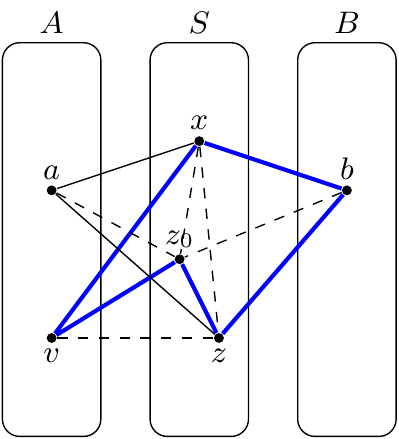}
\caption{Proof of Lemma~\ref{lem:lh:minsep-xab}.}
  \label{fig:L42}
\end{center}
\end{figure}

\begin{lemma}\label{lem:lh:pmc-dom-indset1}
Let $G$ be a long-hole-free graph, let $\pmc$ be a PMC in $G$, and let $M \subseteq \pmc$ be an independent set.
Then either $|M|=1$ and  $\pmc \subseteq N[M]$, or
  there exists $D \in \cc(G-\pmc)$ with $M \subseteq N(D)$.
\end{lemma}
\begin{proof}
  Assume that the first alternative does not hold. That is,
  if $|M|=1$ then $\pmc \not \subseteq N[M]$.
Then, by Theorem~\ref{thm:pmc}, for every $v \in M$ there exists a component $D \in \cc(G-\pmc)$ covering
a nonedge from $v$ to some other vertex of $\pmc$.
Let $D \in \cc(G-\pmc)$ maximize $|N(D) \cap M|$. Suppose, contrary to the second alternative, that $M \nsubseteq N(D)$ and let $t \in M \setminus N(D)$; see Figure~\ref{fig:L43}.
Denote $\compfam = \{D' \in \cc(G-\pmc)~|~t \in N(D')\}$; note that $D \notin \compfam$.
For every $D' \in \compfam$ let $f(D') = N(D') \cap N(D) \cap M$. 
Since $M$ is an independent set, for every $x \in M \cap N(D)$ there exists $D_x \in \compfam$ satisfying $\{x,t\}\subseteq N(D_x)$.
In particular, $\bigcup_{D' \in \compfam} f(D') = N(D) \cap M$. 
On the other hand, for every $D' \in \compfam$, the maximality of
$|N(D) \cap M|$ and the
fact that $t \in N(D')$ imply that $f(D') \neq N(D) \cap M$.
Consequently, there exist two components $D_1,D_2 \in \compfam$ with inclusion-wise incomparable $f(D_1)$ and $f(D_2)$.
Let $x_i \in f(D_i) \setminus f(D_{3-i})$ for $i=1,2$ and let $P_i$ be a shortest path from $x_i$ to $t$ via $D_i$. 
Let $Q$ be a shortest path from $x_1$ to $x_2$ via $D$. Then $t-P_1-x_1-Q-x_2-P_2-t$ is a long hole in $G$, a contradiction.
This finishes the proof of the lemma.
\end{proof}

\begin{figure}[tb]
\begin{center}
\includegraphics{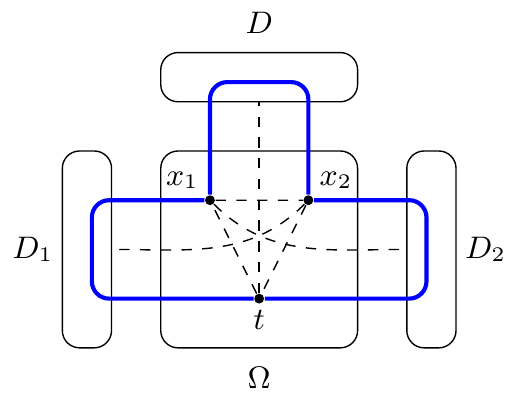}
\caption{Proof of Lemma~\ref{lem:lh:pmc-dom-indset1}.}
  \label{fig:L43}
\end{center}
\end{figure}

\begin{lemma}\label{lem:lh:pmc-dom-nonnei}
Let $G$ be a long-hole-free graph, let $\pmc$ be a PMC in $G$, and let $v \in \pmc$. 
Then either $\pmc \subseteq N[v]$ or there exists $D \in \cc(G-\pmc)$ with $\pmc \setminus N(v) \subseteq N(D)$
(in particular, $v \in N(D)$). 
\end{lemma}
\begin{proof}
  Suppose that $\pmc \not \subseteq N[v]$ but there is no component $D$ as in
  the statement of the lemma; see Figure~\ref{fig:L44}.
Let $A \subseteq \pmc \setminus N[v]$ be inclusion-wise minimal such that there is no component $D \in \cc(G-\pmc)$
with $A \cup \{v\} \subseteq N(D)$. 
Observe that $A$ is non-empty, because there is at least one non-edge within $\Omega$ with one endpoint $v$, so in particular there is at least one component $D\in \cc(G-\pmc)$ satisfying $v\in N(D)$.
If $A$ is an independent set, then so is $A \cup \{v\}$, and then Lemma~\ref{lem:lh:pmc-dom-indset1} contradicts
the choice of $A$. Hence, there exists an edge $xy \in E(G[A])$. 
By the minimality of $A$, there exist a component $D_x \in \cc(G-\pmc)$
with $(A \setminus \{y\}) \cup \{v\} \subseteq N(D_x)$
and a component $D_y \in \cc(G-\pmc)$
with $(A \setminus \{x\}) \cup \{v\} \subseteq N(D_y)$.
By the choice of $A$, we have $y \notin N(D_x)$ and $x \notin N(D_y)$. 
Let $P_x$ be a shortest path from $x$ to $v$ via $D_x$ and similarly define $P_y$ in $D_y$.
But then $x-P_x-v-P_y-y-x$ is a long hole in $G$, a contradiction.
\end{proof}

\begin{figure}[tb]
\begin{center}
\includegraphics{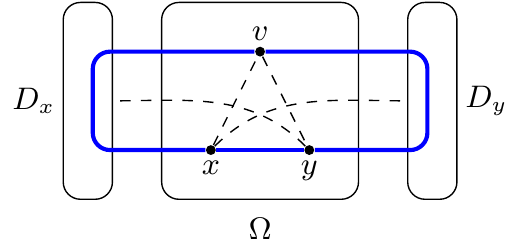}
\caption{Proof of Lemma~\ref{lem:lh:pmc-dom-nonnei}.}
  \label{fig:L44}
\end{center}
\end{figure}


\begin{lemma}\label{lem:lh:pmc-dom}
Let $G$ be a long-hole-free graph, let $\pmc$ be a PMC in $G$, and let $v \in \pmc$ be arbitrary.
Then either $\pmc \subseteq N[v]$ or 
there exist two vertices $x \in N(v) \setminus \pmc$ and $y \in N(v)$ such that $\pmc \subseteq N[\{v,x,y\}]$.
\end{lemma}
\begin{proof}
Assume that $\pmc \not\subseteq N[v]$.
Let $D \in \cc(G-\pmc)$ be a component with $\pmc \setminus N(v) \subseteq N(D)$ (it exists by Lemma~\ref{lem:lh:pmc-dom-nonnei}). 
Recall that by Proposition~\ref{prop:pmcnbrs} $N(D)$ is a minimal separator and $D$ is a full component for
$N(D)$. By Lemma~\ref{twofull} there exists $B \neq D$ such that $B$ is a full
component for $N(D)$.
By Lemma~\ref{lem:lh:minsep-xab}, there exist $x \in N(v) \cap D$ and $y \in N(v) \cap B$ with 
$N(D) \subseteq N[v] \cup N(x) \cup N(y)$. Since $\pmc\setminus N(v) \subseteq N(D)$, we have
$\pmc \subseteq N[\{v,x,y\}]$ as desired.
\end{proof}

We can now deduce the following:

\begin{theorem}\label{thm:3dom}
For every long-hole-free graph $G$ and for every potential maximal clique $\pmc$ in $G$
there exists a set $Z \subseteq V(G)$ of size at most three such that $\pmc \subseteq N[Z]$.
\end{theorem}

Theorem~\ref{thm:3dom} immediately implies Theorem~\ref{thm:seps} using standard
techniques.
\begin{proof}[Proof of Theorem~\ref{thm:seps}.]
Let $G$ and $\weight$ be as in the statement of the theorem. By Theorem~\ref{thm:3dom}, it suffices to show that there is a potential maximal clique in $G$ that is a balanced separator with respect to $\weight$.
To this end, let $F$ be a minimal chordal completion of $G$.
A folklore result (see e.g.~\cite{GrzesikKPP19}) is that $G+F$ admits a tree decomposition where the bags are exactly the maximal cliques of $G+F$. 
Let $T$ be the tree of the decomposition. For every edge $e \in E(T)$, let $T_1^e$ and $T_2^e$ be the two components of $T-\{e\}$ and for $i=1,2$ let $V_i^e$ be the union of all the bags of $T_i^e$.
Orient the edge $e$ from the endpoint in $T_i^e$ with smaller weight of $V_i^e$ to the one with the larger weight, breaking ties arbitrarily. 
Let $t \in V(T)$ be a node of zero outdegree. Then it can be easily checked that the bag at $t$ is a balanced separator and we are done.
\end{proof}

\bibliographystyle{abbrv}

\bibliography{../references}

\end{document}